\documentclass[pra,twocolumn,nofootinbib]{revtex4-2}
\usepackage{mathtools}
\usepackage{tikz-cd} 
\usepackage{pst-node}
\usepackage{adjustbox}
\usepackage[utf8x]{inputenc}     
\usepackage{enumerate}
\usepackage{bm}  
\usepackage{float} 
\usepackage{wasysym} 
\usepackage{mathrsfs} 
\usepackage[colorlinks=true,     linkcolor=blue, citecolor=blue, filecolor=blue, urlcolor=blue]{hyperref}
\usepackage{amsfonts} 
\usepackage{amsbsy} 
\usepackage{amscd} 
\usepackage{pstricks} 
\usepackage{multirow} 
\usepackage{tikz}
\usepackage{color}
\usepackage{slashed}
\usepackage{array}
\usepackage{tablefootnote}

\usetikzlibrary{arrows,positioning,shapes.geometric} 
\usepackage[compat=1.1.0]{tikz-feynman}          
\usepackage{slashed}
\usepackage{centernot}
\usepackage{multirow}
\usepackage{tabularx}
\usepackage{soul}  
\usepackage{listings} 
\newif\iflocaldraftfigures
\localdraftfiguresfalse
\iflocaldraftfigures
\renewcommand{\includegraphics}[2][]{\fbox{\scriptsize\ttfamily\detokenize{#2}}}
\fi
\usepackage{color}

\usepackage{amsthm}
\newtheorem{theorem}{Theorem}
\newtheorem{definition}{Definition}
\newtheorem{corollary}{Corollary}
\newtheorem{lemma}{Lemma}

\newtheorem{proposition}{Proposition}

\newcommand{\G}{\mathcal{G}}

\newcommand{\M}{\mathcal{M}} 
\newcommand{\x}{\mathbf{x}}

\newcommand{\mf}[1]{\mathfrak{#1}}

\newcommand{\mc}[1]{\mathcal{#1}} 
\newcommand{\mbf}[1]{\mathbf{#1}}
\newcommand{\y}{\mathbf{y}}

\DeclareMathOperator{\im}{Im}

\usepackage{makecell}
\usepackage{pifont}
\usepackage{amsthm}

\newcommand{\ket}[1]{|#1\rangle}
\newcommand{\bra}[1]{\langle#1|}
\newcommand{\braket}[2]{\langle#1|#2\rangle}
\newcommand{\quantexp}[2]{\langle#2|#1|#2\rangle}

\newcommand{\GL}[2]{\text{GL}(#1,#2)}

\newcommand{\rank}{\text{rank}}
\usepackage{amsthm}
\usepackage{booktabs}

\usepackage{mathtools}

\DeclarePairedDelimiter\floor{\lfloor}{\rfloor}
\newcommand{\vtheta}{{\bm\theta}}
\newcommand{\valpha}{{\bm\alpha}}
\newcommand{\vbeta}{{\bm\beta}}
\newcommand{\vvarphi}{{\bm\varphi}}
 
\usetikzlibrary{quantikz}

\renewcommand{\selectlanguage}[1]{}

\begin{document} 
	\title{
		Equivalence of maximal and generic reachability for non-universal Variational Quantum Circuits
	} 
	
	\author{Vishal~S.~Ngairangbam}
	\email{vishal.ngairangbam@kit.edu}
	
	\author{Michael~Spannowsky} 
	\email{michael.spannowsky@kit.edu}
	\affiliation{Institute for Theoretical Physics, Karlsruhe Institute of Technology, \\ Wolfgang-Gaede-Str. 1, 76131 Karlsruhe, Germany}
	\affiliation{Institute for Quantum Materials and Technologies, Karlsruhe Institute of Technology, \\ Karlsruhe 76131, Germany}
	\begin{abstract}  
		Employing problem-specific non-universal Variational Quantum Circuits aligned with a suitable state preparation has become a standard approach to counter the difficulty in training universal ans\"atze.  However, due to their non-universality, diagnosing their reachability and trainability remains reference-state-specific. In this work, we establish the equivalence of maximal reachability and generic reachability over arbitrary reference states as a consequence of the \emph{principal orbit-type theorem}. Thereafter, assuming that the global minimum of the cost function is achieved on a subset that can be described as the image of a smooth function, we derive necessary and sufficient conditions for non-zero probability of reachability under generically sampled reference states. Furthermore, when the solution set is assumed to be realised through a real analytic map that embeds a solution manifold, we show that local surjectivity is generically obtained on the entire solution manifold if it is attained at a single point. As a practical design rule, the need for local surjectivity automatically translates to a necessary dimensional criterion: reachability requires the target set's topological dimension to be at least as large as the co-dimension of the generic orbit. 
		Numerical simulations show that finite-depth optimisation consistently fails when the dimensional obstruction applies, while unobstructed cases show improved convergence.

	\end{abstract}
	
	\maketitle

	\section{Introduction}
	
	The difficulty in training~\cite{McClean2018,Bittel2021,Anschuetz2022,Larocca2025BPreview} universal Variational Quantum Circuits (VQCs) that can represent any special unitary transformation acting on the $n$-qubit quantum system has instigated the use of non-universal VQC designs that can represent unitaries within proper subgroups. However, non-universal VQC designs that can represent exponentially large subgroups of the special unitary group still show barren plateaus~\cite{Cerezo2021cost,Larocca2022diagnosing,Ragone2024,Fontana2024} in the loss landscape with polynomially scaling ones showing better landscapes. An important and complementary requirement of such non-universal ans\"atze, is to quantify their state reachability~\cite{Akshay2020,Haug2021,Singh:2025evh,Oh2026} and hence, suitability for different tasks and characterisation of possible quantum advantages~\cite{Goh:2023kcm,Cerezo2025}. Due to the non-transitive nature of the unitary group action, such analyses are generally reference-state dependent, and many analyses primarily restrict to an invariant subspace of the representation~\cite{Zeier2011,DAlessandro2021,Albertini2023,Wiersema2020HVA,Gard2020,Meyer2023,Sauvage2024} of the dynamical Lie algebra and do not generalise~\cite{Monbroussou:2023syc}.

	\begin{figure*}[t]
		\includegraphics[width=0.7\textwidth]{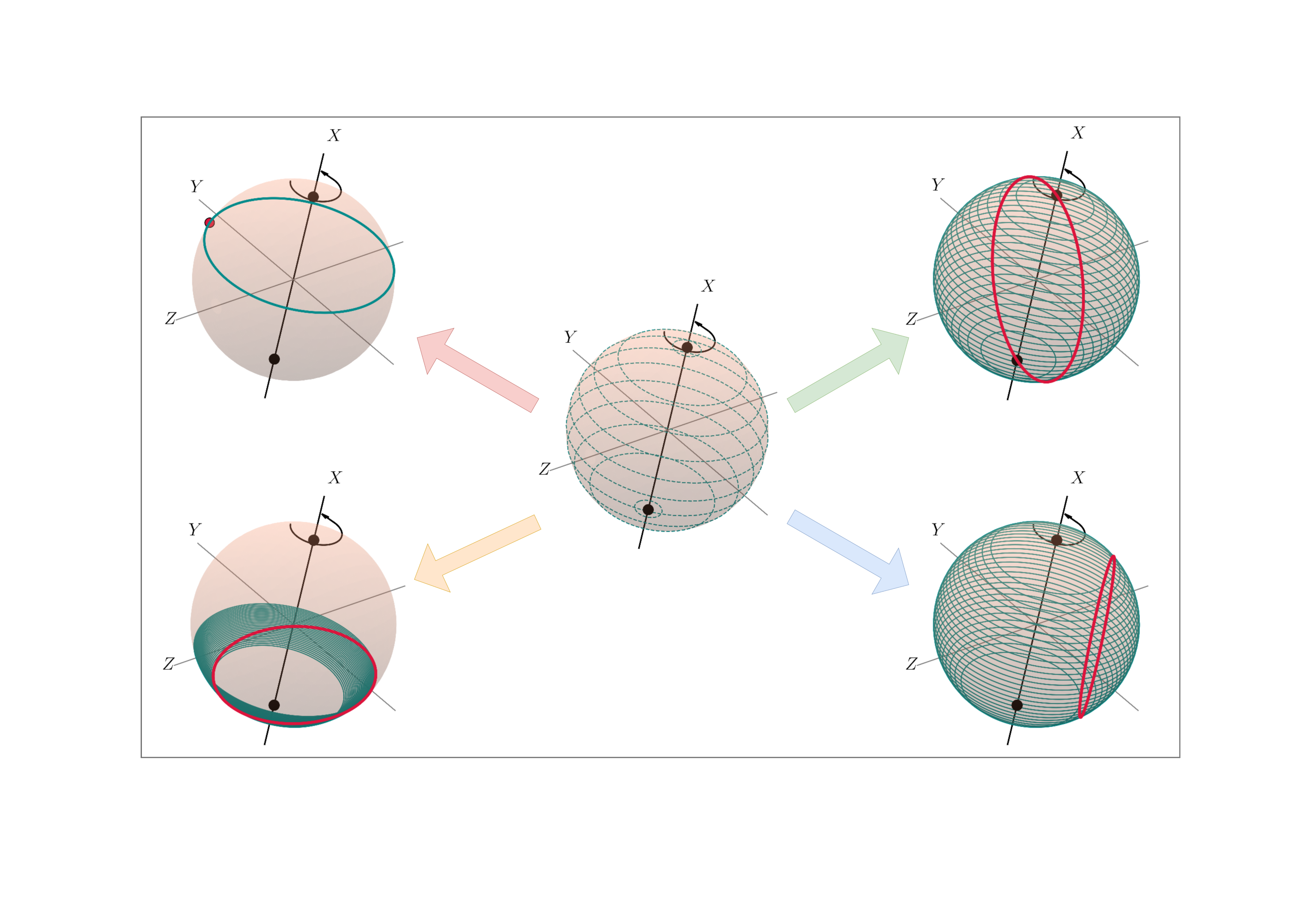}
		\caption{Stratification of the sphere $S^2$ under $SO(2)$ rotations about the $X$ axis and its relation to different orientations of the solution manifold $\M^*$. At the top left, $\M^*$ is a single point on $S^2$ and the initial point must lie on the exact orbit. When $\M^*$ is a circle, the range of orbits that intersect it depends on its relative orientation with the rotation. The area of the union of such orbits is non-zero as long as there is one point of intersection where the tangent to the circle and an orbit are not parallel. In this work, we prove that this intuitive picture carries over to compact $\G$-actions on $\M_N$ by utilising the principal orbit-type theorem. 
		}
		\label{fig:bloch_strat} 
	\end{figure*}
	In this work, by virtue of the \emph{principal orbit-type theorem}~\cite{princ_orb_smooth,BredonGlenE1972Itct,tomDieck1987,Duistermaat2000} and the resulting \emph{orbit-type} stratification of the pure state manifold, we establish that maximal reachability of non-universal VQC is achieved for generically prepared reference states.  Similar orbit-type stratifications are well-established in \emph{gauge theory} literature~\cite{Rudolph:2002fe,Hertsch:2010zz,Rudolph:2017pug}.
	A simple illustration with $SO(2)$ rotations about the $X$-axis on a sphere is shown in Fig.~\ref{fig:bloch_strat} (centre). The generic circular orbits cover the entire sphere except for the two antipodal points where the $X$-axis intersects the sphere. These two antipodal points remain fixed under the rotation and have $SO(2)$ as their \emph{isotropy subgroup}. Note that the union of all circular orbits has a surface area equal to the sphere, since the two points have zero area (i.e. a \emph{null} set). 
	
	The same stratification of the pure state manifold into generic and non-generic orbits arises, \emph{mutatis mutandis}, for any smooth action of a compact  Lie group. An immediate conclusion is that \emph{maximally reachable orbits in the pure state manifold are the generic ones}.  Thereby, assuming that the solution set can be parametrised by the image of a \emph{smooth map}, we derive necessary and sufficient conditions so that there is a non-zero probability of reaching the solution manifold under a generically sampled reference state. A heuristic explanation based on the sphere is outlined in Fig.~\ref{fig:bloch_strat}. A natural corollary is that the dimensions of the solution manifold must at least match the codimensions\footnote{In this work, by $\dim X$, we always refer to the real topological dimension of the object $X$. For any $X$ that is a subset of $Y$, its codimension with respect to $Y$ is simply $\dim Y-\dim X$.} of the principal orbits. Additionally, assuming that the solution manifold can be realised as a real analytic embedding, we show that the regular component (i.e. union of all orbits of maximum topological dimensions) either intersects it generically or it is entirely contained in the singular component. Numerically evaluating the dynamical Lie group action on the ground state manifolds of Hamiltonians for small qubit counts, shows agreement with this behaviour.

	The rest of the paper is structured as follows. We first discuss some recent and important related works in Sec.~\ref{sec:related_work}. This is followed by building up the necessary mathematical prerequisites in Sec.~\ref{sec:prelims}. Sec.~\ref{sec:gen_geom} highlights the main theoretical insights of the paper by utilising the principal orbit-type theorem. In Sec.~\ref{sec:num_sim}, we present numerical simulation to corroborate the underlying theory.  We summarise and conclude in Sec.~\ref{sec:conc}. A brief review of additional mathematical preliminaries is provided in Appendix~\ref{app:background}, and all proofs are collected in Appendix~\ref{app:proofs}.

	\section{Related Work} 
	\label{sec:related_work}
	Geometric analysis of controllability and reachability is well-explored in quantum control~\cite{khaneja00geometric,y8nw-n1wb,opt_cont_geom_rev}. More recently, complementary to the algebraic analysis of reachability and trainability, there is a growing body of literature that looks into such aspects from the geometric perspective. In ref.~\cite{Wiedmann:2025zvz}, convergence  to the ground state of Hamiltonians is discussed for universal and non-universal VQCs by primarily considering the group manifold geometry. Additionally, as expected when restricted to a proper subgroup, reachability is no longer guaranteed from generic states. Our work addresses complementary aspects by considering the looser criterion of probabilistic reachability of the ground state under generic state preparation and its relation to overparametrisation. 
	
	Relatedly, ref.~\cite{Margeta-Cacace:2026rss} analyses the geometry of VQA reachability by approximating the reachable submanifold of the group by Krylov-Lie algebras and Lie groups. Our work focusses on the reachable subset on the state manifold. While this is induced by the group geometry through the action, it is equivalent (i.e. diffeomorphic) if and only if the initial state's isotropy subgroup is trivial. Interestingly, theorem 3.1.1 is a direct analogue on the group manifold of proposition~\ref{prop:toy_reachability} relating the non-zero reachability of a finite volume to rank saturation of the Quantum Fisher Information Matrix.

	\section{Preliminaries}
	\label{sec:prelims}
	\subsection{Setup} 
	Consider Variational Quantum Circuits (VQCs) acting on an $n$-qubit system with parametrised gates built out of a sequence $\mbf{G}=(G_1,G_2,...,G_k)$ of unique generators $G_j$ of the form $W_l(\vvarphi)=\prod_{j=1}^l\,e^{i\, \varphi_j \,G_j}$, $\vvarphi\in\mathbb{R}^l$, for $l\leq k$. These are sequentially applied to form a VQC with $d$ total parametrised gates  as 
	\begin{equation}
		\label{eq:u_vqc} 
		U(\vtheta)=\prod_{a=0}^{q+1}\, W_{l_a}(\vvarphi_a)\quad,
	\end{equation}where $q=\floor{\frac{d}{k}}$, $l_a=k$ for $a\leq q$, $l_{q+1}=d-qk$ and $\vtheta=\oplus_{a=0}^{q}\,\vvarphi_a\in\mathbb{R}^d$.
	Unless otherwise specified, each gate's generator will be regarded as a single Pauli string. Since each gate parametrises a one-parameter subgroup of the compact group $SU(N)$, they are periodic and the parameters can be restricted to a subset $\Theta\subset\mathbb{R}^d$ formed by restricting each gate parameter to its principal range.  
	If $\ket{\psi_r}\in\mc{H}_N\cong\mathbb{C}^{N}$, $N=2^n$, is the initial reference state, then for any $\vtheta\in \mathbb{R}^d$, the evolved state is 
	\begin{equation}
		\label{eq:hilbert_action}
		\ket{\psi(\vtheta)}=U(\vtheta)\,\ket{\psi_{r}}\quad.
	\end{equation}
	By eliminating the global phase ambiguity of $\ket{\psi_r}$, these unitaries equivalently act smoothly on the manifold of pure states $\M_N\cong\mathbb{C}P^{N-1}$, that consist of rays $$\psi\equiv\{e^{i\lambda}\ket{\psi}:\lambda\in\mathbb{R}\}$$ for each representative normalised state $\ket{\psi}\in\mathcal{H}_N$.
	
	The cost function $C(\vtheta)$ is defined in conjunction with a measurement protocol through which the parameters $\vtheta$ are updated via some optimisation algorithm. We leave the exact definition of the cost function arbitrary, but assume that the set of solutions in $\M_N$ can be parametrised by a smooth map $f:\M^*\to\M_N$. Thus, all states in $\im(f)\subset\M_N$ evaluate to the global minimum of $C(\vtheta)$ in $\M_N$.

	\subsection{Dynamical Lie group action}
	
	The range of elements $g\in SU(N)$ represented by $U(\vtheta)$ for $\vtheta\in\Theta$ increases with increasing $d\in\mathbb{N}$. The subgroup $\G$ that contains all possible elements represented by $U(\vtheta)$ as $d\to\infty$, is formed by exponentiating the Dynamical Lie Algebra (DLA) $\mf{g}$, of the generators in $\mbf{G}$. It is constructed by extending the generator set such that the commutators between all generators (including the ones not initially in $\mbf{G})$ are closed. Their span, therefore, forms a subalgebra $\mf{g}$  of $\mf{su}(N)$.

	Let $\mc{U}(\G)$ denote the diffeomorphic image of $\G$ in $\GL{N}{\mathbb{C}}$ corresponding to the representation of $\G$ parametrised by the VQC. Thus, at any arbitrary $d$, the unitary implemented by a VQC is a map $U:\Theta\to\mc{U}(\G)$ so that $\im(U)\subset\mc{U}(\G)$. One therefore has the $\G$-action $\mbf{A}:\G\times\M_N\to\M_N$  on $\M_N\cong\mathbb{C}P^{N-1}$, induced by the corresponding action of $\mc{U}(\G)$ on $\mathcal{H}$. For the ray $\psi_r$, the \emph{orbit map} is defined as $\mbf{A}^{\psi_r}:\G\to\M_N$ with the assignment $g\mapsto \mbf{A}(g,\psi_r)$. The range of rays that can be traversed under the unitary evolution (eq.~\ref{eq:hilbert_action}) is therefore contained in the orbit $\mc{O}_{\psi_r}\subset\M_N$ of $\psi_r$. We can define the parameter-to-orbit map $\Gamma^{\psi_r}:\Theta\to\mc{O}_{\psi_r}$ which describes the geometry of the evolution for a reference ray $\psi_r$. Note that the Lie group $\mc{G}$ is a closed subgroup of $SU(N)$, and therefore \emph{compact}. Moreover, since it consists of  exponentiated Lie algebra elements of $\mf{g}$, it is \emph{connected}. Thus, $\mc{O}_{\psi_r}$ is a connected, embedded submanifold of $\M_N$.

	\subsection{Local overparametrisation on orbits}
	
	Utilising the implicit function theorem, 
	for any state $\psi_r\in\M_N$, states $\psi(\valpha)$ in an open neighbourhood around $\psi(\mbf{0})=\psi_r$ in $\mc{O}_{\psi_r}$ can be described  in terms of local coordinates $\valpha\in\mathbb{R}^{\dim\mc{O}_{\psi_r}}$. Defining the operator $\partial_{\alpha_i}\equiv\frac{\partial}{\partial\alpha_i}$, the components of the Fubini Study metric on $\mc{O}_{\psi_r}$ are given as
	\begin{equation}
		\begin{split}
			[\mbf{g}(\valpha)]_{ij}=\text{Re}&[ \braket{\partial_{\alpha_i}\psi(\valpha)}{\partial_{\alpha_j}\psi(\valpha)}
			\\
			&-\braket{\partial_{\alpha_i}\psi(\valpha)}{\psi(\valpha)}\;\braket{\psi(\valpha)}{\partial_{\alpha_j}\psi(\valpha)}]\quad.
		\end{split} 
	\end{equation} 
	Under the parameter-to-orbit map $\Gamma^{\psi_r}:\Theta\to\mc{O}_{\psi_r}$,  rays in the neighbourhood will be accessed by changes in the parameters $\vtheta$, and therefore we will have $\valpha(\vtheta):\Theta\to\mathbb{R}^{\dim\mc{O}_{\psi_r}}$. 
	Up to a factor of four, the Quantum Fisher Information Matrix (QFIM) is the pullback of $\mbf{g}(\valpha)$ to the parameter manifold
	\begin{equation*}
		\begin{split} 
			\mbf{F}(\vtheta_0)= 4\;\mbf{J}^T_{\psi_r}(\vtheta_0)\;\mbf{g}(\valpha)\;\mbf{J}_{\psi_r}(\vtheta_0)\quad,
		\end{split} 
	\end{equation*}
	\noindent 
	where $\mbf{J}_{\psi_r}(\vtheta_0)$ is the Jacobian representation of differential push-forward $d\Gamma^{\psi_r}_{\vtheta_0}:T_{\vtheta_0}\Theta\to T_{\psi(\vtheta_0)} \mc{O}_{\psi_r}$, $\psi(\vtheta_0)=\Gamma^{\psi_r}(\vtheta_0)$. The positive definiteness of the Riemannian metric $\mbf{g}(\valpha)$ therefore gives 
	\begin{equation*}
		\rank(\mbf{F}(\vtheta_0))=\rank(\mbf{J}_{\psi_r}(\vtheta_0))\quad.
	\end{equation*}
	The dimensions of $\mbf{J}_{\psi_r}(\vtheta_0)$, i.e. $\dim\mc{O}_{\psi_r}\times d$, where $\dim\mc{O}_{\psi_r}$ denotes the intrinsically accessible degrees of freedom naturally sets an upper bound on the QFIM rank~\cite{Larocca:2021jub}
	$$R(\psi_r)\equiv\max_{d\in\mathbb{N}}\left(\max_{\vtheta\in\Theta}(\rank(\mbf{F}(\vtheta)))\right)=\dim\mc{O}_{\psi_r}\quad.$$
	For some fixed $d$, the VQC is said to be overparametrised when $\max_{\vtheta\in\Theta}(\rank(\mbf{F}(\vtheta))=\dim\mc{O}_{\psi_r}$. In view of the preceding discussions, we see that this is a strictly local phenomenon in the orbit $\mc{O}_{\psi_r}$. 
	\subsection{Reachability under exact orbit knowledge}
	As a precursor to reachability of the solution manifold under generic state preparation, let us first work out the analogous reachability of a unique solution ray $\psi_*$, assuming an exact knowledge of its orbit $\mc{O}_{\psi_*}$, i.e. $\mc{O}_{\psi_r}=\mc{O}_{\psi_*}$. 
	Suppose the reference ray $\psi_r$ is sampled uniformly with respect to the Riemannian volume form in $\mc{O}_{\psi_r}$. In terms of local intrinsic coordinates, it is given as 
	\begin{equation*}
		d\Omega_{\mc{O}_{\psi_r}}=\sqrt{\det{(\mbf{g}(\valpha))}}\; d\alpha_{1}\wedge d\alpha_{2}\wedge ...\wedge d\alpha_{\dim\mc{O}_{\psi_r} }\quad,
	\end{equation*} 
	
	For any representative target state, $\ket{\psi_*}\in\psi_*\in\mc{T}_*$, the set of reference states in $\mc{H}_N$ that can reach $\ket{\psi_*}$ can be obtained from the unitary action
	\begin{equation*}
		\ket{\psi_r}=U^\dagger(\vtheta)\;\ket{\psi_*}\quad,
	\end{equation*}
	Thus, on the orbit $\mc{O}_{\psi_r}=\mc{O}_{\psi_*}$, we have the analogous map $\Gamma^*:\Theta\to\mc{O}_{\psi_r}$  from the parameter space to the orbit manifold. Under a uniform probability measure with respect to $d\Omega_{\mc{O}_r}$, the probability that a fixed $\psi_*$ is reachable is therefore
	\begin{equation}
		P\left(\psi_*\in\im(\Gamma^{\psi_r})\right)=\frac{1}{V}\int_{\im(\Gamma^*)} \;d\Omega_{\mc{O}_{\psi_r}}\quad,
	\end{equation}
	where $V=\int_{ \mc{O}_{\psi_r}}\,d\Omega_{\mc{O}_{\psi_r}}$. 
	
	The difference between overparametrised and underparametrised regimes has so far been looked into the loss landscape of the parameter space~\cite{Larocca:2021jub,pmlr-v139-you21c} and the fact that the QFIM rank is generic~\cite{Monbroussou:2023syc} on the connected domain $\Theta$ due to the analyticity of $\Gamma^{\psi_r}:\Theta\to\mc{O}_{\psi_r}$. Utilising the smoothness alone, we can distinguish overparametrised and underparametrised reachability on the orbit $\mc{O}_{\psi_r}$ :   
	\begin{proposition}
		\label{prop:toy_reachability} 
		If the reference state $\ket{\psi_r}$ is prepared in such a way that the ray $\psi_r\ni\ket{\psi_r}$ has uniform probability under the normalised Riemannian volume form $d\Omega_{\psi_r}$ on the orbit $\mc{O}_{\psi_r}\ni\psi_*$, then the probability that $\im(\Gamma^{\psi_r})$ contains $\psi_*$ is non-zero if and only if  $\max_{\vtheta\in\Theta}\rank(\mbf{F}^*(\vtheta))=\dim\mc{O}_{\psi_r}$, where $\mbf{F}^*(\vtheta)$ is the Quantum Fisher Information Matrix of the map $\Gamma^*:\Theta\to\mc{O}_{\psi_r}$.  
	\end{proposition} 
	\noindent 
	The forward implication follows from the strict positivity of the Riemannian volume form, while the reverse implication follows from criticality of the image for non-surjective smooth maps via Sard's theorem~\cite{Sard1942-yn,Lee2012}. A detailed proof is provided in Appendix~\ref{app:toy_reachability}. Since the genericity of rank on connected domains for analytic maps and the criticality of the image for non-surjective smooth maps will be utilised repeatedly, a more detailed explanation is provided in Appendix~\ref{app:gen_ranks}.

	Recalling that the QFIM rank is generic but not constant in $\Theta$, it is important to distinguish the ranks of  $\mbf{F}^*(\vtheta)$ and $\mbf{F}(\vtheta)$, and the dependence of the proposition on maximality of $\rank(\mbf{F}^*(\vtheta))$. This is because for a particular $d$ and $U(\vtheta)$ the circuit can have $\rank(\mbf{F}(\vtheta))=\dim\mc{O}_{\psi_r}$ for almost every $\theta\in\Theta\subset\mathbb{R}^d$ but still have $\rank(\mbf{F}^*(\vtheta))<\dim\mc{O}_{\psi_r}$.  As a concrete example on the Bloch sphere (i.e. $\mathbb{C}P^1$), consider ${\psi^*}=\{e^{i\lambda}\ket{0}:\lambda\in\mathbb{R}\}$, and the unitary $U(\theta_y,\theta_z)=e^{i\,\frac{\sigma_z}{2}\theta_z}\;e^{i\,\frac{\sigma_y}{2}\theta_y}$. We see that $\max_{\vtheta\in\Theta}\rank(\mbf{F}^*(\vtheta))=1$ while $\max_{\vtheta\in\Theta}\rank(\mbf{F}(\vtheta))=2$. Therefore, the set of reference states that can reach $\psi_*$ under $U(\theta_y,\theta_z)$ corresponds to those with $\theta_y=0$, which is the great circle passing through the $Z$-axis and therefore has zero Fubini-Study volume on $\mathbb{C}P^1$.

	As a result of the proposition, overparametrisation is therefore necessary but not sufficient for a non-zero probability of reaching a unique target state that is contained in $\mc{O}_{\psi_r}$. Heuristically, this is because one can begin to \emph{quantify} reachable volumes in the orbit only after rank saturation of $\mbf{F}(\vtheta)$. Unless~\cite{Onishchik1963,1220755} $\mf{g}\cong\mf{su}(N)$ or $\mf{g}\cong\mf{sp}(N/2)$, the $\G$-action is not transitive and $\mc{O}_{\psi_r}$ is a closed embedded submanifold of $\M_N$ and is strictly lower dimensional $\dim\mc{O}_{\psi_r}<\dim\M_N$. Additionally, this can happen even when $\dim\mf{g}>\dim\M_N$, making it a stringent and exact dimensional bound on asymptotic reachability as $d\to\infty$. Thus, the \emph{maximal orbits} in $\M_N$, quantify the maximum reachability of any non-universal VQC. As we shall see, these maximal orbits are encountered for generically prepared states making it relatively easy to access the maximum reachability of non-universal VQCs. 
	\section{Generic geometry on the pure state manifold} 	
	\label{sec:gen_geom}
	In evaluating the reachability probability within a single orbit, we assumed that there was a single unique state. From a mathematical perspective, this was highly structured and it is not in general given, that we know the exact orbit of a target state. In this section, we scrutinise such behaviour by looking into the volume of states in the pure state manifold that can reach a set of targets. This is made possible by the principal orbit-type theorem and the resulting genericity of the principal orbits. 
	
	\subsection{Genericity of maximal orbits} 
	
	In general, any orbit $\mc{O}_{\psi}$ of a smooth Lie group action is a homogeneous space and inherits a smooth structure from $\G$ via the quotient $\G/\mc{S}_{\psi}$, where $\mc{S}_{\psi}$ is the isotropy subgroup of $\mc{O}_{\psi}$. For any two distinct orbits $\mc{O}_{\psi_1}$ and $\mc{O}_{\psi_2}$, if $\mc{S}_{\psi_1}$ and $\mc{S}_{\psi_2}$ are conjugate, i.e. $\mc{S}_{\psi_1}=g\,\mc{S}_{\psi_2}g^{-1}$ for some $g\in\G$, then the two orbits are naturally diffeomorphic. Thus, the isotropy subgroup (up to conjugacy) determines the \emph{orbit type}. More importantly each \emph{connected} subset of $\M_{(\mc{S})}$ consisting of all points in $\M_N$ that have a conjugate copy of $\mc{S}$ as its isotropy group forms an embedded submanifold, and these components form a \emph{stratification} of $\M_N$ called the \emph{orbit-type stratification}. Consequently, (i) there is a maximal orbit type $\mc{O}_0\cong\G/\mc{S}_0$ called the \emph{principal orbit-type} characterised by $\mc{S}_0$ the smallest possible isotropy subgroup called the principal isotropy type, (ii) the union of all such orbits in $\M_N$, say $\M_0$, forms an \emph{open, dense}\footnote{Recall that under a topology $\tau$ on a set $X$, a set $D\subset X$ is dense in $X$ if and only if every non-empty open neighbourhood $V\in\tau$ contains a non-empty intersection with $D$.} subset of $\M_N$ and (iii) the complement $\bar{\M}_0=\M_N\setminus\M_0$ is a strictly lower dimensional subset of measure zero in the Riemannian volume form in $\M_N$. Thus, for the Fubini-Study volume form $d\Omega_{\rm FS}$, $\Omega_{\rm FS}(\bar{\M}_0)=\int_{\bar{\M}_0}\;d\Omega_{\rm FS}=0$ and $\Omega_{\rm FS}(\M_0)=\int_{\M_0}\;d\Omega_{\rm FS}=1$. Therefore, under the application of deep VQCs of the form $U(\vtheta)$, the generically encountered orbit-type is also maximal and therefore
	\begin{equation*}
		R(\psi)=\max_{\psi_r\in\M_N}R(\psi_r)=\dim\mc{O}_0\quad a.e. \text{ in } \M_N\quad.
	\end{equation*}
	Thus, without explicitly evaluating the dynamical Lie algebra, the dimensions of the principal orbit can be evaluated as the saturated rank of a VQC on a generically sampled reference state.

	\subsection{Reachability under generic state preparation} 
	While there is no assured reachability for non-universal ans\"atze, we can now derive necessary and sufficient conditions for nonzero probabilistic reachability of $\im(f)$ under generically sampled states via $d\mu_{\rm FS}$. Define the map $\Sigma:\G\times \M^*\to\M_N$ as the composition $\Sigma=\mbf{A}\circ(\text{id}_\G\times f)$ of the smooth $\G$-action $\mbf{A}:\G\times\M_N\to\M_N$ on $\M_N$. Thus, $\mc{K}_*\equiv\im(\Sigma)$ denotes the subset of $\M_N$ from which the solution manifold $\im(f)$ is reachable under the $\G$-action.  For a non-zero probability of reaching some state in $\im(f)$ under a generically prepared reference state,  $\mc{K}_*$ should have non-zero Fubini-Study measure. 
	
	Define the set  $S_0=\im(f)\cap\M_0$, as the intersection of the image of $f$ with the principal component $M_0$. The following lemma states the local equivalence of surjectivity of the map $\Sigma:\G\times\M^*\to\M_N$ to the restriction of $f$ to the preimage of $S_0$ in $\M^*$, say $\tilde{\Sigma}:\G\times f^{-1}(S_0)\to\M_N$. 
	\begin{lemma}
		\label{lemma:surj_equiv} 
		An open neighbourhood $V$  in $\G\times\M^*$ contains a point $(g,p)$ at which $\rank(d\Sigma_{(g,p)})=\dim(\M_N)$ if and only if $V\cap(\G\times f^{-1}(S_0))$ contains a point $(g',p')$ such that $\rank(d\tilde{\Sigma}_{(g',p')})=\dim(\M_N)$. 
	\end{lemma}
	\noindent 
	The proof follows since $\M_0$ is open and dense in $\M_N$, making $\tilde{\Sigma}$ smooth and every open neighbourhood in $\M_N$ contains an element of $\M_0$, respectively. A more detailed proof is provided in Appendix~\ref{app:surj_equiv}. Motivated by this lemma, we state the following theorem which captures the necessary and sufficient conditions on when $\Omega_{\rm FS}(\mc{K}_*)>0$.
	\begin{theorem}
		\label{theorem:suff_nec} 
		If $d\Omega_{\rm FS}$ is the Fubini-Study measure on $\M_N$, then $\Omega_{\rm FS}(\mc{K}_*)=\Omega_{\rm FS}(\mc{K}_*^0)$, where $\mc{K}_*^0=\mc{K}_*\cap\M_0$. Thereby, $\Omega_{\rm FS}(\mc{K}_*)>0$ if and only if the restricted map $\tilde{\Sigma}:\G\times f^{-1}(S_0)\to\M_N$, where $S_0=\im(f)\cap\M_0$, is locally surjective for at least one element in the domain, i.e. there is some $(g,p)\in\G\times f^{-1}(S_0)$ such that $\rank(d\tilde{\Sigma}_{(g,p)})=\dim\M_N$. 
	\end{theorem}      
	\noindent 
	Heuristically, the proof follows from the equivalence of surjectivity of $\tilde{\Sigma}$ with $\Sigma$ (from the above lemma) and $\Omega_{\rm FS}(\bar{\M_0})=0$. A detailed proof is provided in Appendix~\ref{app:suff_nec}. 
	
	Note that the theorem only concerns the nullity of the set $\mc{K}_*$ under the Riemannian volume form $d\Omega_{\rm FS}$. For conditions on whether $\Omega_{\rm FS}(\mc{K}_{*})=1$, i.e. $\mc{K}_*$ is a set of full measure in $\M_N$ for non-universal VQCs, note that $\dim\mc{O}_0<\dim\M_N$. In this case, knowledge of the dimension alone cannot lead to a statement about $\mc{K}_*$ being a full measure set since it needs to intersect almost every principal orbit in $\mc{M}_0$. However, by additionally assuming that $f:\M^*\to\M_N$ is an analytic embedding, we can further quantify the range of $S_0$ in $\im(f)$.
	
	\subsection{Generic rank on analytic solution manifolds}
	Since the group action is \emph{real analytic}, we can further refine \emph{local surjectivity} on the domain of $\Sigma$ leading to important implications on the co-domain $\M_N$. Analytic solution manifolds are encountered for instance, in variational quantum eigensolvers where the target eigensubspace's degeneracy $m$, makes it $\M^*\cong\mathbb{C}P^{m-1}$. In such cases, due to the additional real analyticity, we can obtain the following stronger result:
	\begin{theorem}
		\label{theorem:gen_analytic}
		Let $f:\M^*\to\M_N$ be the real analytic embedding of a compact connected Riemannian manifold $\M^*$, $R_0\equiv\{ f(p)\in \im(f):\rank(d\Sigma_{(e,p)})=r_{\rm max}\}$, with $r_{\rm max}=\min(\dim\mc{O}_0+\dim\M^*,\dim\M_N)$, and $d\mu_*$ be any probability measure that is absolutely continuous with respect to the induced Riemannian measure on $\im(f)$. 
		\begin{itemize}
			\item if $S_0\neq\emptyset$ then $\mu_*(S_0\cup S_{E})=1$ where $S_{E}=\im(f)\cap\M_{E}$ with $\M_{E}$ is the union of all non-principal orbits of dimensions equal to $\dim\mc{O}_0$ in $\M_N$
			\item if $\dim\M_N<\dim\mc{O}_0+\dim\M^*$ and $R_0\neq\emptyset$ then $\mu_*(R_0)=1$, $\Omega_{\rm FS}(\mc{K}_*)>0$, and $S_0$ is dense in $\im(f)$. 
		\end{itemize}
	\end{theorem}  
	\noindent 
	In addition to the orbit-type stratification, the proof relies on the well known result ~\cite{Mityagin2020} that analytic functions in a connected domain are either identically zero or are non-zero almost everywhere, and the fact that $R_0$ is dense in $\im(f)$  as it is a set of full measure under the strictly positive Fubini-Study measure on $\im(f)$, which thereby guarantees density of $S_0$ by lemma~\ref{lemma:surj_equiv}. 
	A detailed proof provided in Appendix~\ref{app:gen_analytic}.

	The first result quantifies that for an analytically embedded solution manifold, either it falls entirely in a singular subset, or it falls almost entirely in the regular component. This is independent of the local surjectivity of the map $\Sigma:\G\times\M^*\to\M_N$.  The second result states that when there is at least one state $\psi\in\im(f)$ where the map $\Sigma$ becomes locally surjective, then the principal orbits intersect the embedded manifold \emph{transversally}
	\begin{equation*}
		T_{\phi}\,\mc{O}_{0,\phi}+T_{\phi}\,\im(f)=T_{\phi\,}\M_0  
	\end{equation*}
	in a dense subset $\phi\in S_0$ of $\im(f)$. Here, $\mc{O}_{0,\phi}$ are principal orbits but are not necessarily the same for every $\phi$. While the measure of this set will be non-zero in $\im(f)$, it need not be full.

	\subsection{Dimensional obstruction to reachability} 
	While the additional assumptions of the theorem guarantee denseness of the set $S_0$, note that it requires $R_0\neq\emptyset$, beyond $\Delta D\geq0$. Thus, it still requires the exact knowledge of at least one point in the image of $f$, which is not a priori known in practice. To operationalise the theorems in the absence of exact knowledge of $\im(f)$, but only its dimensionality, we consider the sufficient conditions for when $d\Omega_{\rm FS}(\mc{K}_*)=0$. Since this is a statement on whether the set $\mc{K}_*$ has measure zero or not, it can be translated to any measure $d\mu$ on $\M_N$ that is \emph{absolutely continuous} with respect to $d\Omega_{\rm FS}$. 	An absolutely continuous measure $d\mu$ with respect to $d\Omega_{\rm FS}$ simply means that $d\Omega_{\rm FS}(A)=0\implies \mu(A)=0$ for every measurable subset $A\subset\M_N$. Conversely, $d\mu$ assigns non-zero measure to null sets under $d\Omega_{\rm FS}$ if it is not absolutely continuous.

	The requirement of local surjectivity on at least one point directly leads
	to a necessary condition on the dimensions via the differential pushforward $d\Sigma_{(g,p)}:T_g\G\oplus T_p\M^*\to T_{{\Sigma}(g,p)}\,\M_N$. At constant $g$, the second term of the direct sum decomposition pushes forward the manifold $\M^*$ therefore has at most $\dim(\M^*)$ rank. The same for the first term pushes forward the group action and a maximum rank of $\dim(\mc{O}_0)$. Since their image on the co-domain $T_{{\Sigma}(g,q)}\,\M_N$ may overlap forming a positive-dimensional vector subspace, $\rank(d{\Sigma}_{(g,p)})\leq\dim(\mc{O}_0)+\dim(\M^*)$. For local surjectivity we need $\rank(d{\Sigma}_{(g,p)})=\dim(\M_N)$ and therefore $\dim(\M^*)\geq \dim(\M_N)-\dim(\mc{O}_0)$. We get the important contraposition that captures a practical necessary condition of reachability under generic state preparation. 
	\begin{corollary}
		\label{cor:dim_obs} 
		If the reference state $\ket{\psi_r}$ is sampled with a probability measure $d\mu$ that is absolutely continuous with respect to the Fubini-Study measure on $\M_N$,  the probability that the range of reachable states via a Variational Quantum Circuit $U(\vtheta)$ is zero as long as $$\Delta D\equiv\dim\M^*-(\dim\M_N-\dim\mc{O}_0)<0\quad,$$
		where $\mc{O}_0\cong\G/\mc{S}_0$ specifies the principal orbits of the associated dynamical Lie group $\G$ of $U(\vtheta)$.
	\end{corollary}
	\noindent
	Therefore, when $\Delta D<0$, beyond the dimensions of $\im(f)$, one requires additional information of where it is located in $\M_N$ so that the probability of the $\mc{O}_{\psi_r}$ intersecting $\im(f)$ is non-zero. In other words, one needs to prepare $\ket{\psi_r}$ with a distribution $d\mu$ that assigns non-zero probability to lower dimensional subsets $A\subset\M_N$, where $\Omega_{\rm FS}(A)=0$.

	\section{Reachability under sector-transitive actions} 
	In the previous section, we left the DLA, $\mf{g}$, and hence the group $\mc{G}$ arbitrary. In this section, by restricting the generator sequence $\mbf{G}$ to consist of Pauli strings, we derive \emph{sufficient conditions} under which a generically prepared reference state's orbit would \emph{almost surely} contain a non-zero intersection with the solution set $\im(f)$. This is achieved by connecting the manifold geometry to the algebraic structure of the unitary representation and its linear action on the Hilbert space.  
	
	\subsection{Structure of the group action} 
	Let $\mf{R}(\mf{g})$ be the representation of the Lie algebra $\mf{g}$ on $\mf{gl}(N,\mathbb{C})$ consisting of anti-Hermitian matrices. If one can find a normal matrix $C\in\mf{gl}(N,\mathbb{C})$ with $K\geq 1$ distinct eigenvalues that commutes with every element $X\in\mf{R}(\mf{g})$, every $X$ decomposes into a direct sum of the form
	\begin{equation}
		\label{eq:inv_dec} 
		X=\bigoplus_{a=1}^{K} X_{a}\quad,
	\end{equation} 
	where $X_a$ are $d_a\times d_a$ anti-Hermitian matrices acting on some $d_a$-dimensional invariant subspace $\mc{H}_{a}\cong\mathbb{C}^{d_a}$ of $\mc{H}_{N}$, and $\sum_{a=1}^K\, d_a=N$. Therefore, under exponentiation, we have 
	\begin{equation*}
		\exp\left(\bigoplus_{a=1}^K X_a\right)=\bigoplus_{a=1}^K\exp(X_a)\quad.
	\end{equation*}
	If each block faithfully represents some Lie algebra $\mf{g}_a$ which generates a connected Lie group $\G_a$, the group $\G$ decomposes  as the product  
	\begin{equation*}
		\G=\G_1\times\G_2\times...\times \G_K\quad.
	\end{equation*}
	The unitary action of each $\G_a$ is closed on the invariant subspace $\mc{H}_a$.

	To elucidate the induced orbit geometry on each subspace and how it relates to the orbit in $\M_N$, consider the completeness relation of the individual projectors $$\sum_{a=1}^K\;P_a=\;\mathbf{I}_N$$
	on the Hilbert space $\mc{H}_N$, where $P_a=\sum_{j=1}^{d_a}\ket{\Lambda_{j,a}}\bra{\Lambda_{j,a}}$ is the rank $d_a$ projector to each $d_a$-dimensional subspace of eigenvectors. So that for any state $\ket{\psi}\in\mathcal{H}_N$, $\ket{\psi_a}\equiv \sum_{j}\ket{\Lambda_{j,a}}\braket{\Lambda_{j,a}}{\psi}$ is its component in $\mc{H}_a$. Note that the definition of the projector has no phase ambiguities, i.e. for any  $\ket{\Lambda'}=e^{i\lambda}\,\ket{\Lambda}$ we have  $\ket{\Lambda'}\bra{\Lambda'}=\ket{\Lambda}\bra{\Lambda}$ and therefore $\ket{\psi_a}$ does not have a phase ambiguity in its definition once the global phase convention of $\ket{\psi}$ is fixed. Thus, projected states in the subspace $\mc{H}_{a}$ are identified with $p_a\equiv|\braket{\psi_a}{\psi_a}|^2$ and the $\G_a$ action on $\mc{H}_a$ cannot change $p_a$.   Since the phases of these subspaces are physically relevant, the set of states that the unitary evolution can take a state $\ket{\psi_a}$  will be diffeomorphic to an embedded submanifold of $S^{2\,d_a-1}$. However, keeping all phases in the projection will introduce a physically redundant global phase which can be taken care of by taking a quotient of the entire product with $U(1)$.

	The transformations on $\mc{H}_a$ can therefore be equivalently written for each $p_a$ as a smooth $\G_a$-action of the form $\mbf{A}_a:\G_a\times S^{2\,d_a-1}_{p_a}\to S^{2\,d_a-1}_{p_a}$, where $S^{2 d_a-1}_p$  is an extrinsically embedded $2\,d_a-1$ dimensional sphere in $\mathbb{R}^{2\,d_a}$ with radius $\sqrt{p_a}$. Thus, for any state $\ket{\psi}\in\mathcal{H}$, under the invariant sector action of the representation, the orbit decomposes into a product manifold indexed by the weight vector $\mathbf{p}=(p_1,p_2,...,p_K)$ and $\sum_{a=1}^Kp_a=1$, as 
	\begin{equation*}
		\mc{O}_{\psi}(\mbf{p})=\mc{O}_{p_1}\times \mc{O}_{p_2}\times...\times \mc{O}_{p_K}
	\end{equation*}
	where each $\mc{O}_{p_a}$ is diffeomorphic to some submanifold of $S^{2\,d_a-1}$.  Note that while $\mbf{p}\in\mathbb{R}^K$, it has $K-1$ real degrees of freedom and takes up values in the $K-1$-dimensional probability simplex $\Delta^{K-1}$. Additionally, since all $p_a$ need not be non-zero, the factor can contain empty sets which vacuously satisfy the definition of a manifold.

	\subsection{Principal orbit-space under subspace transitive group actions}  
	
	From the preceding discussions, it is clear that a necessary condition for pure rays in $\M_N$ to belong to a principal orbit when all factors $\G_a$ act faithfully on their corresponding subspace $\mc{H}_a$, is for the probability vector $\mbf{p}$ to not have any zero elements. This is because when any $p_a=0$, the isotropy subgroup will contain the factor $\G_a$ which otherwise will be absent when $p_a>0$. If each $\G_a$ action on $S^{2\,d_a-1}_{p_a}$ is transitive, then this is necessary and sufficient. Thus, the principal orbit-space $\M_0/\G$, when each factor group acts transitively on the manifold $S^{2\,d_a-1}$, is diffeomorphic to the interior of the $K-1$ probability simplex $\M_0/\G\cong\Delta^{K-1}_+\cong\mathbb{R}^{K-1}$, where $$\Delta^{K-1}_{+}\equiv\{\mbf{p}\in\Delta^{K-1}:p_a>0\}$$ is the interior of the probability simplex $\Delta^{K-1}$. The principal component $\M_0$ can then be written as a diffeomorphism
	\begin{equation*}
		\M_0\cong \Delta^{K-1}_{+}\times \left(S^{2\,d_1-1}\times...\times S^{2\,d_K-1}\right)/U(1)
	\end{equation*}
	where the $U(1)$ quotient takes care of the redundant global phase. A straightforward dimensional counting 
	\begin{equation*}
		\dim\M_0=K-1+\sum_{a=1}^K (2\,d_a-1)-1=2(N-1)
	\end{equation*}
	reveals the equality of dimensions with $\M_N$. The principal orbits are therefore  $$\mc{O}_0\cong\left(S^{2\,d_1-1}\times...\times S^{2\,d_K-1}\right)/U(1)\quad,$$
	and have codimension $\dim\M_N-\dim\mc{O}_0=K-1$.

	Some possible $\G_a$ (and consequently $\mf{g}_a$) that have transitive group actions on $S^{2\,d-1}$ can be ascertained through the homogeneous space constructions:
	\begin{equation*}
		S^{2d-1}\cong \frac{U(d)}{U(d-1)}\cong \frac{SU(d)}{SU(d-1)}
	\end{equation*}
	for any $d$ and 
	\begin{equation*}
		S^{4m-1}\cong \frac{Sp(m)}{Sp(m-1)}\quad
	\end{equation*}
	for even $d=2m$ where $Sp(m)$ is the compact symplectic group. Thus, when each factor in $\G$ corresponds to one of these groups, the principal orbit space will be diffeomorphic to $\Delta^{K-1}_+$. While we primarily utilise this diffeomorphism for ground state reachability, it could in principle be utilised to engineer efficient search of the entire $\M_N$ with non-universal VQCs since $\M_0$ itself is dense in $\M_N$.

	\subsection{Almost sure reachability of ground states} 
	The previous characterisation of the principal orbit for subspace transitive $\G$  actions can now be utilised to derive sufficient conditions on when the ground state manifold can be reached from almost every state in $\M_N$ for a non-universal VQC. Recalling that under the smooth Lie group action on a compact manifold, the principal component is a \emph{fibre bundle}\footnote{In other words, $\M_0$ is the total space, $\M_0/\G$ is the base space and $\G/\mc{S}_0\cong\mc{O}_0$ is the typical fibre, and one therefore has a smooth projection $\pi:\M_0\to\M_0/\G$ from the total space to the base space whose typical fibre in $\M_0$ are the principal orbits.}, we can define the smooth projection $\pi:\M_0\to\M_0/\G$ and let $\M^*_0\equiv f^{-1}(S_0)$, so that we have the restriction map (see theorem~\ref{theorem:suff_nec}) $f_0:\M^*_0\to\M_0$. The map $\tilde{f}_0:\M^*_0\to\M_0/\G$ defined through the composition $\tilde{f}_0=\pi\circ f_0$ is therefore naturally smooth. Clearly, independent of the exact structure of the group and its relation to invariant subspaces, a sufficient condition so that $\Omega_{\rm FS}(\mc{K}_*)=1$ is that $\tilde{f}_0$ is surjective to $\M_0/\G$. However, without the additional structure on $\G$ and hence $\M_0$ and $\M_0/\G$ it is not so straightforward to  characterise surjectivity or a notion of a volume measure on the co-domain of the map $\tilde{f}_0:\M^*_0\to\M_0/\G$.

	Consider that the ground state subspace $\mc{H}_0$ has a non-trivial intersection with each invariant subspace $\mc{H}_a$, i.e. each $\mc{H}_a\cap\mc{H}_0$ forms a closed Hilbert subspace of $\mc{H}_N$ of positive complex dimension. Thus, there is at least one normalised vector $\ket{\psi_{*,a}}$ in  each $\mc{H}_a$ such that $\quantexp{H}{\psi_{*,a}}=E_0$ and over $a\in\{1,2,...,K\}$ they form a mutually orthonormal set. Therefore, the linear combination
	\begin{equation*}
		\ket{\psi_*}=\sum_{a=1}^{K}\, \sqrt{p_a}\,e^{i\lambda_a}\ket{\psi_{*,a}}\quad
	\end{equation*}
	with arbitrary values of $\mbf{p}=(p_1,p_2,...,p_{K-1})\in\Delta^{K-1}$, $p_K=1-\sum_{a=1}^{K-1} p_a$, and any $\lambda_a\in\mathbb{R}$ is also a ground state with $\quantexp{H}{\psi_*}=E_0$. Since this is valid for all values of $p\in\Delta^{K-1}$, it is valid on the interior $\Delta^{K-1}_+$ which parametrises the principal orbits. Thus, $\tilde{f}_0:\M^*_0\to\Delta^{K-1}_+$ is surjective and therefore $\Omega_{\rm FS}(\mc{K}_*)=1$.

	While the measure one statement does not assume any additional structure on the Hamiltonian $H$, note that  we require $m\geq K$ for an $m$-fold degenerate ground state. Since $\dim\mathbb{C}P^{m-1}=2(m-1)$, this is stronger than simply requiring $\dim\M^*\geq\dim\M_0/\G=\dim\M_N-\dim\mc{O}_0$ for the non-zero measure statement.  Additionally, the requirement of having a non-zero intersection with the ground state eigenspace is highly specialised and reflects a complete transfer of \emph{a priori} knowledge from the state preparation to the VQC ansatz design. Nevertheless, this provides a useful design principle  when we can enumerate the entire ground state sector's relation to the invariant subspace decomposition of a DLA.

	\section{Numerical Simulations} 
	\label{sec:num_sim}
	In this section, we show numerical support for the theoretically derived results in the paper by constructing four DLA groups and their corresponding VQC architecture of the form Eq.\ref{eq:u_vqc}. Note that the theoretical results are on reachability of the group and require infinite depth as well as the absence of any obstructions in optimisation. Therefore, they test the practical implications of the results but work under strictly weaker assumptions. Nevertheless, we find finite depth optimisation consistently fails when the ground state manifold is dimensionally obstructed, while show good convergence for most sampled reference states in the unobstructed case barring the DLA $\mf{g}_{(XX,YY)}$ for $n=7$.

	\subsection{Details of Dynamical Lie Algebras} 
	The sequence $\mathbf{G}$ is constructed with generators in the Pauli string basis with at most two non-identity factors. In the main text, as well as here, we therefore use the conventions 
	\begin{equation*}
		\begin{split}
			\sigma^i_s&\equiv\underbrace{I_2\otimes I_2\otimes...\otimes\underbrace{\sigma_s}_{i^{th}\text{ qubit}}\otimes I_2...\otimes I_2}_{n\text{ terms}}\quad,\\
			\sigma^i_{s_1}\sigma^j_{s_2}&\equiv\underbrace{I_2\otimes ...\otimes\underbrace{\sigma_{s_1}}_{i^{th}\text{ qubit}}\otimes...\otimes\underbrace{\sigma_{s_2}}_{j^{th}\text{ qubit}}\otimes ...\otimes I_2}_{n\text{ terms}}\quad,
		\end{split}
	\end{equation*} 
	for $s\in\{X,Y,Z\}$, and utilise the sets $\mf{C}_1=\{1,2,...,n\}$ and the choice of \emph{ordered} two-combinations $\mf{C}_2=\{(a,b)|a,b\in\mf{C}_1:a<b \}$. 
	The four classes of DLAs are constructed as : (i) $\mf{g}_{(XY)}$: $\sigma^i_X\sigma^j_Y$ for $(i,j)\in\mf{C}_2$, (ii) $\mf{g}_{(XX,YY)}$: $\sigma^i_X\sigma^j_X$ and $\sigma^i_Y\sigma^j_Y$ for $(i,j)\in\mf{C}_2$, (iii) $\mf{g}_{(X,YY)}$: $\sigma^i_X$ for $i\in\mf{C}_1$  and $\sigma^i_Y\sigma^j_Y$ for $(i,j)$ in $\mf{C}_2$, and (iv) $\mf{g}_{(Z,YY)}$: the same as $\mf{g}_{(X,YY)}$ with $\sigma^i_Z\in\mf{C}_1$ instead of $\sigma^i_X$.  For these choices of $\mbf{G}$, generators of the DLA are evaluated with the \texttt{liealg.lie\_closure} function implemented in \texttt{PennyLane (v0.42.3)}~\cite{Bergholm:2018cyq}.  Note that $\mf{g}_{(X,YY)}$ and $\mf{g}_{(Z,YY)}$ are isomorphic Lie algebras, while $\mf{g}_{(XX,YY)}$ is a subalgebra of both.
	
	\subsection{Principal orbit dimensions} 
	
	\begin{table}[t]
		\centering
		\small
		\setlength{\tabcolsep}{5pt}
		\renewcommand{\arraystretch}{1.5}		
		\begin{tabular}{lrr|rr|rr|rr}
			\toprule
			DLA &\multicolumn{2}{c|}{$n=4$}&\multicolumn{2}{c|}{$n=5$}&\multicolumn{2}{c|}
			{$n=6$}&\multicolumn{2}{c}{$n=7$}    \\
			\cline{2-3}\cline{4-5}\cline{6-7}\cline{8-9}
			& $D_\mathfrak{g}$ & $D_{0}$ & $D_{\mathfrak{g}}$ & $D_0$ &  $D_{\mathfrak{g}}$ & $D_0$ & $D_{\mathfrak{g}}$ & $D_{0}$ \\
			\midrule
			$\mathfrak{g}_{(XY)}$ & $\mathbf{16}$ & $\mathbf{16}$ & 64 & 48 & 256 & 112 & 1024 & 240 \\
			$\mathfrak{g}_{(XX,YY)}$ & 60 & 27 & 255 & 59 & 1020 & 123 & 4095 & 251 \\
			$\mathfrak{g}_{(Z,YY)}$ & 126 & 29 & 510 & 61 & 2046 & 125 & 8190 & 253 \\
			$\mathfrak{g}_{(X,YY)}$ & 126 & 29 & 510 & 61 & 2046 & 125 & 8190 & 253 \\
			\bottomrule
		\end{tabular}
		\caption{$D_{\mf{g}}\equiv\dim\mf{g}$ and principal orbit dimensions $D_0\equiv\dim\mc{O}_0$ for four choices of $\mbf{G}$ for $n\in\{4,5,6,7\}$. We see that except for $\mf{g}_{(XY)}$ and $n=4$, all other choices have a positive dimensional principal isotropy subgroup, so that $D_0$, rather than $D_\mf{g}$ determines maximal asymptotic reachability of non-universal VQCs.}
		\label{tab:dla_orb_dim}
	\end{table}
	Since $\dim\mc{O}_0$ is the rank of the orbit map for generic states in $\M_N$, we implemented the $\G$-action obtained by exponentiating elements of the DLA and its linear action on normalised pure states $\ket{\psi}$ in $\mathbb{C}^N$. This was converted to the presentation of an \emph{orbit map}  on $\M_N$ as follows. We first recast the complex vector $\mbf{z}\in\mathbb{C}^N$ in terms of \emph{affine coordinates} as $\mbf{z}'=(z_2/z_1,z_3/z_1,...,z_N/z_1)\in\mathbb{C}^{N-1}$, where $z_i$ are the components of $\mbf{z}$. $\mbf{z}'$ is further represented as $\mbf{x}=\text{Real}(\mbf{z}')\oplus\text{Imag}(\mbf{z}')$ so that $\mbf{x}\in\mathbb{R}^{2(N-1)}$ parametrises a local chart of $\M_N\cong\mathbb{C}P^{N-1}$ as a real analytic manifold. The Jacobian and its rank are evaluated numerically at the identity element of $\G$ for thirty randomly sampled states using \texttt{Jax (v0.6.2)}~\cite{jax2018github}. As expected, the rank was constant over all sampled states.

	For $n\in\{4,5,6,7\}$ and the four DLAs the principal orbit dimensions $\dim\mc{O}_0\equiv D_0$ and $\dim\mf{g}\equiv D_\mf{g}$ are shown in table~\ref{tab:dla_orb_dim}. We see that only $\mf{g}_{(XY)}$ for $n=4$ has a locally free action on the principal orbits, while the rest have $\dim\mf{g}>\dim \mc{O}_0$ so that $\dim\mc{S}_0>0$. Additionally, for the rest $\dim\mf{g}>\dim\M_N$ but $\dim\mc{O}_0<\dim\M_N=2^{n+1}-2$ highlighting the crucial role played by $\mc{O}_0$ rather than $\mf{g}$ in characterising \emph{generic and maximal state reachability} of non-universal VQC ans\"atze.

	\begin{table}
		\centering
		\small
		\setlength{\tabcolsep}{3pt}
		\renewcommand{\arraystretch}{1.75}	
		\begin{tabular}{lcc|cccc}
			\toprule
			DLA& $n=4$ & $n=6$ &\multicolumn{2}{c}{$n=5$}&\multicolumn{2}{c}{$n=7$} \\
			\cline{4-5}
			\cline{6-7}
			&$D_{\rm GS}$ &$D_{\rm GS}$ & $D_{\rm GS}$&$R_{d\Sigma}$ & $D_{\rm GS}$& $R_{d\Sigma}$\\  
			\midrule
			$\mathfrak{g}_{(XY)}$ & 7 & 31 & 48 & 54 & 240 & 246 \\
			$\mathfrak{g}_{(XX,YY)}$ & 6 & 30 & 59 & $\mbf{62}$ & 251 & $\mbf{254}$ \\
			$\mathfrak{g}_{(X,YY)}$ & 14 & 62 & 61 & $\mbf{62}$ & 253 & $\mbf{254}$ \\
			$\mathfrak{g}_{(Z,YY)}$ & 14 & 62 & 61 & $\mbf{62}$ & 253 & $\mbf{254}$ \\
			\bottomrule
		\end{tabular}
		\caption{Orbit dimensions $\dim\mc{O}_{\rm GS}\equiv D_{\rm GS}$ of the ground state manifold of the Hamiltonian in Eq.~\ref{eq:ham} for different $n$. For odd $n$, we also report $R_{d\Sigma}\equiv \rank(d\Sigma_{(g,p)})$ and highlight those that satisfy the second condition of theorem~\ref{theorem:gen_analytic} ($R_{d\Sigma}=\dim\M_N=2^{n+1}-2$).
		}
		\label{tab:ground_geom} 
	\end{table}
	
	\subsection{Generic rank on ground state manifolds}  
	We test the numerical rank conditions of the ground state manifold, and dimensional obstruction in finding the ground state of the Hamiltonian:
	\begin{equation}
		\label{eq:ham} 
		\begin{split} 
			H=\sum_{k=1}^2J_k\sum_{\{i,j\}\in E_k} &(\sigma_X^i\sigma_X^j+\sigma_Z^i\sigma_Z^j)\quad.
		\end{split} 
	\end{equation} 
	Here, $E_1$  and $E_2$ are edge sets consisting of the next and next-to-nearest neighbour connections, respectively, under periodic boundary conditions, and $J_1=1.0$ and $J_2=0.6$. The cost function is the energy expectation $C(\vtheta)=\quantexp{H}{\psi(\vtheta)}$.
	For odd $n$, the ground state is four-fold degenerate, while it is non-degenerate for even $n$. Therefore, for odd $n$, we have $\M^*\cong\mathbb{C}P^{3}$ with $\dim(\M^*)=6$, for even $n$ we have $\M^*\cong\mathbb{C}P^0$ with $\dim(\M^*)=0$. 
	We use four different classes of VQCs with DLAs: $\mf{g}_{(XY)}$,
	$\mf{g}_{(XX,YY)}$, $\mf{g}_{(X,YY)}$ and 
	$\mf{g}_{(Z,YY)}$. Details of their construction and training are given in the supplementary material. 
	
	In table~\ref{tab:ground_geom}, we show the orbit dimensions of the ground state. For even $n$, it is the Jacobian rank of the orbit map of the unique state while for odd $n$, it is evaluated on hundred randomly sampled states in $\im(f)$. It also shows the rank of $d\Sigma:\G\times\M^*\to\M_N$ for odd $n$, which was evaluated for randomly sampled group elements for each state in $\im(f)$. We see that the ground states reside in singular orbits for even $n$ while they are regular for odd ones. Moreover, from theorem~\ref{theorem:gen_analytic}, we see that the generic rank of $d\Sigma_{(g,p)}$ is saturated at $\dim\M_N$ for all DLAs except $\mf{g}_{(XY)}$ for $n\in\{5,7\}$.

	\subsection{Dimensional obstruction to convergence}

	\begin{table}
		\centering 
		\small	
		\resizebox{0.49\textwidth}{!}{
			\begin{tabular}{lrcc|rcc}
				\toprule
				DLA&$\Delta D$ & $n=4$ & $n=6$& $\Delta D$ & $n=5$ & $n=7$ \\[5pt]
				\midrule
				$\mathfrak{g}_{(XY)}$ & -14 & $1.000^{+0.096}_{-0.037}$ & $1.256^{+0.057}_{-0.031}$ & -8 & $0.277^{+0.015}_{-0.061}$ & $0.602^{+0.090}_{-0.102}$ \\[5pt]
				$\mathfrak{g}_{(XX,YY)}$ & -3 & $1.080^{+0.226}_{-0.044}$ & $1.387^{+0.020}_{-0.042}$ & 3 & $0.000^{+0.000}_{-0.000}$ & $0.073^{+0.019}_{-0.015}$ \\[5pt]
				$\mathfrak{g}_{(Z,YY)}$ & -1 & $0.469^{+0.059}_{-0.021}$ & $0.844^{+0.030}_{-0.026}$ & 5 & $0.000^{+0.000}_{-0.000}$ & $0.005^{+0.004}_{-0.001}$ \\[5pt]
				$\mathfrak{g}_{(X,YY)}$ & -1 & $0.497^{+0.068}_{-0.073}$ & $0.926^{+0.067}_{-0.053}$ & 5 & $0.000^{+0.000}_{-0.000}$ & $0.003^{+0.006}_{-0.001}$ \\[5pt]
				\bottomrule
			\end{tabular}
		}
		\caption{
			$\operatorname{med}(\Delta E_r)$, with asymmetric
			error bars evaluated with the upper and lower quartiles shown as the
			superscript and subscript, respectively for
			\emph{overparameterized} VQCs, evaluated by training over ten randomly sampled
			reference states. $\Delta D\geq0$ is a necessary condition for probabilistic reachability under corollary~\ref{cor:dim_obs}. 
		}
		\label{tab:codim_vqe_summary}
	\end{table}
	For each DLA, we construct an over-parametrised VQC such that its QFIM rank saturates for generic reference states. The redundancy for odd $n$ is kept to a single layer above rank saturation while for even $n$ we add an additional layer. These VQCs are trained ten times on randomly initialised reference states (sampled with a measure absolutely continuous with respect to $d\Omega_{\rm FS}$) with randomly sampled initial weights for each state. 
	
	The random reference states were sampled by first initialising a $2N$-dimensional real vector with elements uniformly sampled in $[-1,1)$. This is then converted to an $N$-dimensional complex vector after which it is normalised. For each parametrised gate, the weights are uniformly sampled in $[-\pi,\pi)$.  
	
	All VQCs were trained for a maximum of 4500 iterations with an early stop criterion if the cost function converged to within a difference of $10^{-8}$ with the true ground state energy.\footnote{All numerical experiments were done in double precision by utilising the \texttt{jax.numpy.float64} and \texttt{jax.numpy.complex128} data types for real and complex variables, respectively.}  The first half of these iterations utilised the \texttt{Adam}~\cite{AdamOpt} optimiser implemented in  \texttt{optax (v0.2.8)}~\cite{deepmind2020jax} with a learning rate of 0.01. After this, we implemented a quantum natural gradient descent with the same initial learning rate which decays by a factor of 0.1 per three hundred iterations until it reaches $10^{-8}$. The training generally converged within single precision during the first phase for $\mf{g}_{(X,YY)}$ and $\mf{g}_{(Z,YY)}$ at $n=5$.

	For each run $r$, we evaluate 
	\begin{equation*}
		\Delta E_r=\frac{\min_{j\in\mc{N}}(C_r(\vtheta_{j})-E_0)}{E_1-E_0}\quad,
	\end{equation*} 
	where $E_0$ and $E_1$ are the ground state and first excited state energies, respectively, and $\mc{N}$ runs over the total number of iterations in the training. In table~\ref{tab:codim_vqe_summary}, we show  the median $\operatorname{med}(\Delta E_r)$ with asymmetric errors evaluated from the upper and lower quartiles along with explicit values of $\Delta D$ grouped into even and odd qubits. For even $n$ where the principal orbits of all four DLAs have the dimensional obstruction, we see poor convergence. The same persists for odd where $\mf{g}_{(XY)}$ is still dimensionally obstructed. For the ones that satisfy the second condition of theorem~\ref{theorem:gen_analytic}, we see that apart from $\mf{g}_{(XX,YY)}$ for $n=7$, all other cases have converged to the ground state within less than a percent of the gap $E_1-E_0$.

	While the good convergence for those that satisfy $\Delta D>0$ for ten generically sampled states may be due to the ground state manifold having non-trivial intersection with each invariant sector, we find that it nevertheless requires overparametrisation. 
	To do this, we fix a generic reference state and train the VQC ten times from random weight initialisation for different $\Delta\bar{d}_0\in\mathbb{Z}$ such that $d=d_0+\Delta\bar{d}_0\;L$, $d_0$ being the threshold of rank saturation for the principal orbits. For $n=5$, we set $L=10$ and $-5\leq \Delta\bar{d}_0\leq 5$, while for $n=7$, $L$ is the size of $\mbf{G}$, and $-3\leq \Delta\bar{d}_0\leq 3$. 
	For each $\Delta \bar{d}_0$, we plot the median over the ten training instances and the interquartile range as the error bars in Fig.~\ref{fig:vqe_conv_odd} from which it is evident that overparametrisation is necessary for convergence. 
	\begin{figure}
		\includegraphics[width=0.35\textwidth]{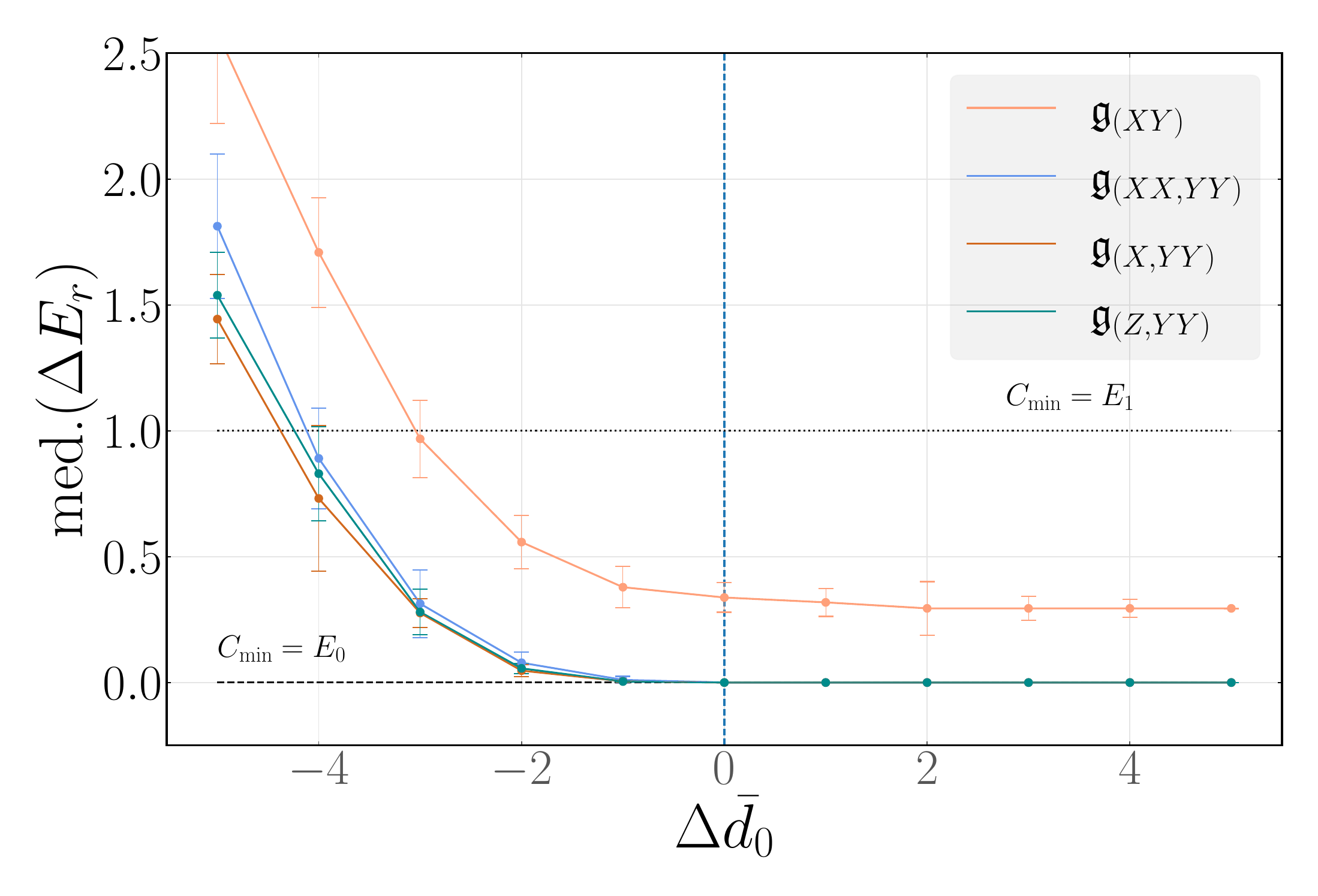}
		\includegraphics[width=0.35\textwidth]{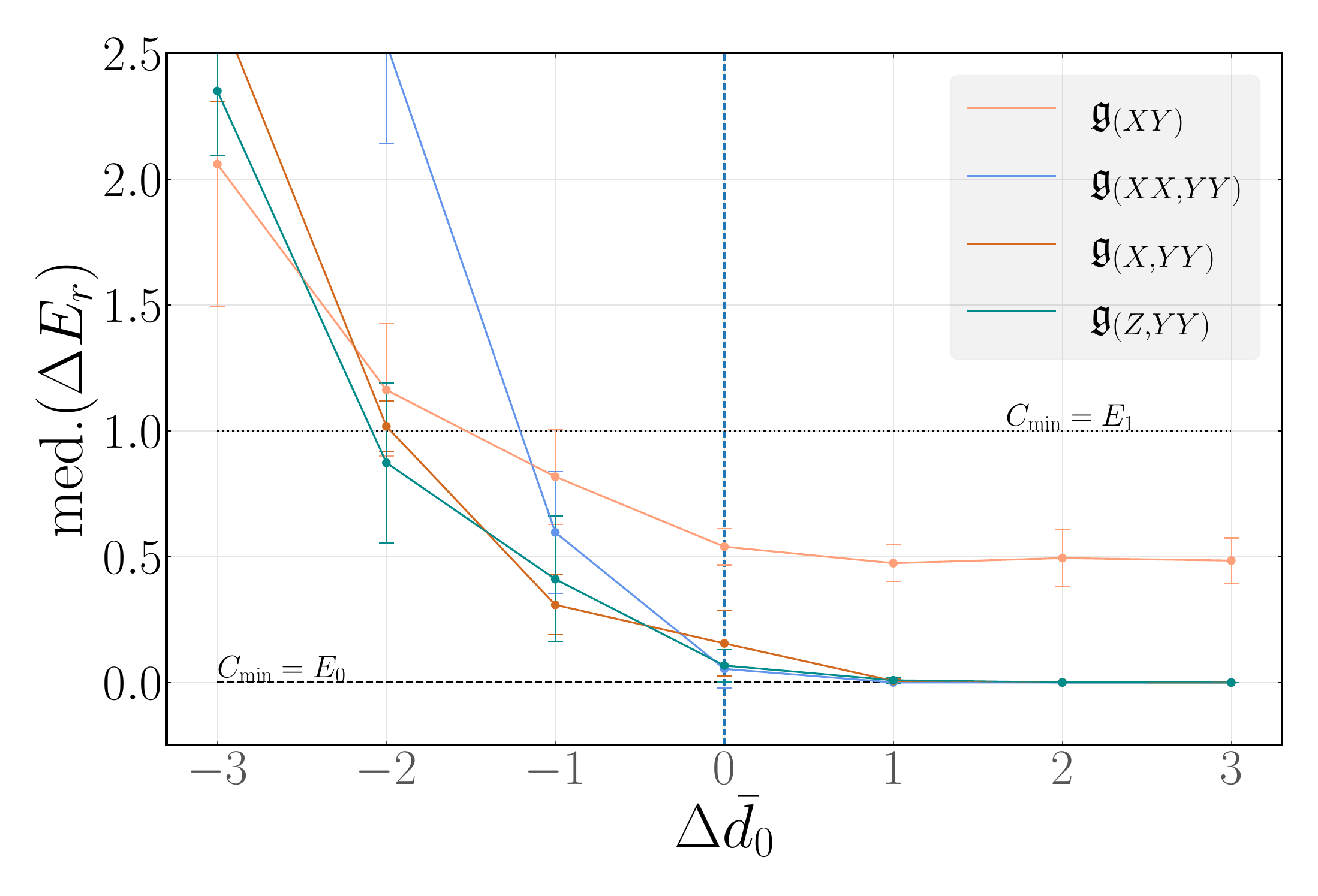}
		\caption{Median of the $\Delta E_r$ and its interquartile range as error bars over ten randomly initialised training runs on a fixed random reference state for $n=5$ (top) and $n=7$ (bottom) verifying that overparametrisation is necessary for improved convergence where $\Delta \bar{d}_0=0$ denotes the threshold.}
		\label{fig:vqe_conv_odd} 
	\end{figure}

	\section{Conclusions}
	\label{sec:conc}
	In this work, we described stratification of the pure state manifold under the smooth action of compact dynamical Lie groups. Due to the principal orbit-type theorem, we showed that generic orbits are the maximal ones. Thus, the dimension of the dynamical Lie algebra is the upper bound of reachability for any pure state only when it acts locally freely on the generic orbits, highlighting the importance of these orbits for characterising the trainability and reachability of non-universal variational quantum circuits. 
	
	Utilising the same theorem, we derived necessary and sufficient conditions on the structure of the group action for a non-zero probability of reaching the solution manifold for generically prepared states. Thereby, we get a simple necessary condition $\Delta D\geq 0$, where $\Delta D$ is the difference between the solution manifold dimension and the codimension of the principal orbits in the pure state manifold. This provides a heuristic design rule in cases where the manifold dimension is generally known from the symmetry of the problem. Numerical experiments for the variational quantum eigensolver for low qubit counts corroborate the utility of $\Delta D$ as a useful check for the need to sample non-generic reference states. 
	
	The work shows that the principal orbits, which are intrinsically determined by the dynamical Lie group action, provide the generic background on which non-universal variational quantum circuits act. Therefore, complementary to the characterisation of DLAs~\cite{Wiersema2024DLA,Allcock2026dynamicallie,Aguilar2024,Kokcu2022,Goh2025}, a concrete characterisation of the differential structure of these principal orbits will help understand the training dynamics and reachability-diagnostics of VQCs in the NISQ era and beyond.

	\appendix

	\section{Background}  
	\label{app:background}
	We give a brief summary of the important results utilised in proving the results.
	
	\subsection{Genericity of rank on domain and image} 
	Many of the theoretical results in this work connect measures on the domain and the co-domain through the analyticity and smoothness of the different functions. Note that analyticity is a stronger condition and therefore subsumes smoothness. The need for local surjectivity for non-zero volume directly results from Sard's theorem~\cite{Sard1942-yn}. Here, we discuss the version applied to smooth manifolds~\cite{Hirsch1976,Lee:2025ppl}.

	Let $\M_1$ and $\M_2$ be two smooth manifolds of real dimensions $m_1$ and $m_2$, respectively.  
	For a smooth map $\Phi:\M_1\to \M_2$, at any point $\x_0\in \M_1$, we can define the differential pushforward $d\Phi_{\x_0}: T_{\x_0} \M_1\to T_{\Phi(\x_0)}\, \M_2$. In terms of coordinate charts $\valpha\in\mathbb{R}^{m_1}$ and $\vbeta(\valpha)\in\mathbb{R}^{m_2}$, this is realised by the Jacobian matrix $$[\mbf{J}(\x_0)]_{ij}=\left.\frac{\partial\beta_i}{\partial\alpha_j}\right|_{\x=\x_0}\quad,$$
	whose rank determines whether the map $\Phi$ is locally surjective or not at $\x_0$. An element $\y\in \M_2$ 
	is a \emph{regular value} if $\rank(\mbf{J}(\x))= m_2$ for all $\x$ in the preimage $\Phi^{-1}(\y)$ of $\y$. If at least one $\x$ in $\Phi^{-1}(\y)$ has $\rank(\mbf{J}(\x))<m_2$, then $\y$ is called a \emph{critical value} of the map $\Phi$. We state Sard's theorem for smooth manifolds along with its immediate consequence required for the following proofs.
	\begin{theorem}[Sard's theorem]
		\label{theorem:sard} 
		For any smooth map $\Phi:\M_1\to \M_2$, the set of critical values in $\M_2$ is a set of measure zero. Therefore, if $\max_{\x\in\M_1}\rank(d\Phi_\x)< m_2\equiv\dim \M_2$, every point in $\im(\Phi)$ is a critical value and therefore, $\im(\Phi)\subset\M_2$ is a set of measure zero. 
	\end{theorem}
	\noindent 
	Note that any $\y\in\M_2\setminus\im(\Phi)$, $\Phi^{-1}(\y)=\emptyset$, and $\mbf{y}$ is vacuously a regular value. Therefore, when $\max_{\x\in\M_1}\rank(d\Phi_\x)=m_2$, the image has non-zero measure but the theorem does not guarantee that it is a set of full measure in $\M_2$. For such stronger results, one can take recourse to analyticity of the map, and the following lemma:
	\label{app:gen_ranks}
	\begin{lemma}
		\label{lemma:analyt_gen_rank}
		Let $\mbf{J}(\x)$ be a matrix whose entries are real analytic functions of $\x$ in a connected domain $\M_1$, let $r(\x)=\rank(\mbf{J}(\x))$ and $r_0=\max_{\x\in \M_1}\rank(\mbf{J}(\x))$, then $r(\x)=r_0$ almost everywhere in $\M_1$.  
	\end{lemma}
	\begin{proof} 
		For $\rank(\mbf{J}(\x_0))=r_0$ at some $\x_0\in \M_1$, recall that there should be at least one $r_0\times r_0$ square sub-matrix of $\mbf{J}(\x)$ with non-zero determinant. Due to
		analyticity, the determinant of this minor is an analytic function of
		$\x$ in the connected domain $\M_1$. Therefore, if it is non-zero at a
		single point, then it is non-zero almost everywhere in
		$\M_1$~\cite{Mityagin2020}.
	\end{proof} 
	The lemma proves the genericity of rank of a matrix with real analytic components on a connected domain and is independent of whether this rank is maximal or not. Applied to an embedding map $\Phi:\M_1\to\M_2$, that is also real analytic, we can transfer the measure zero condition on the domain to its $\im(\Phi)$. However, the measure of $\im(\Phi)$ in $\M_2$ cannot be inferred without an explicitly realised embedding map $\Phi:\M_1\to\M_2$. 
	\subsection{Orbit-type stratification} 
	We discuss essential aspects of the \emph{orbit-type stratification}. Our presentation will primarily be based on Refs.~\cite{Alexandrino2015-ye,BredonAll}. Let $\mbf{A}:\G\times\M_N\to\M_N$ be the smooth unitary action of a compact Lie group $\G$ on $\M_N\cong\mathbb{C}P^{N-1}$ for $N=2^n$. We can first define a principal orbit as
	\begin{definition}
		An orbit $\mc{O}_{\psi}$ is a principal orbit if there is a neighbourhood $V$ of $\psi$ in $\M_N$ such that for every $\psi'\in V$,  $S_{\psi}\subset S_{\mbf{A}(g,\psi')}$ for some $g\in\G$.
	\end{definition}
	From the definition it is evident that the principal orbits are the largest possible orbits in the neighbourhood as their isotropy subgroups are essentially the smallest. We state the principal orbit theorem, together with additional results which will be utilised in the proof:
	\begin{theorem}[Principal orbit-type theorem]
		Under the smooth action of a compact Lie group $\G$ on the manifold $\M_N$,  if $\M_0$ denotes the set of all points in $\M_N$ that are contained in principal orbits, then the following holds:
		\begin{enumerate}
			\item $\M_0$ forms an open, dense subset in $\M_N$
			\item the subset $\mc{M}_0/\mc{G}$ of $\M_N/\G$ is connected
			\item $\M_0$ is connected if $\G$ is connected.  
		\end{enumerate} 
	\end{theorem}

	\subsection{Absolute continuity of measures} 
	Next we define absolute continuity of measures:
	\begin{definition}
		For two measures $d\mu_1$ and $d\mu_2$ defined on the same measurable space $(X,\sigma)$, $d\mu_2$ is absolutely continuous with respect to $d\mu_1$ if for every measurable subset $A\in\sigma$ of $X$, $\mu_1(A)=0\implies\mu_2(A)=0$.
	\end{definition}
	\noindent 
	Primarily for sets of full measure, it is straightforward to deduce from finite additivity of measures that:
	\begin{lemma}
		\label{lemma:conull_abs_cont}
		For two normalised measures $d\mu_1$ and $d\mu_2$ on the measurable space $(X,\sigma)$, if $\mu_1(A)=1$ for any $A\in\sigma$, and $d\mu_2$ is absolutely continuous with respect to $d\mu_1$, then $\mu_2(A)=1$. 
	\end{lemma}
	\begin{proof}
		For any measurable subset $A\subset X$, we have $ A\cap(X\setminus A)=\emptyset$ and $A\cup (X\setminus A)=X$. Since measures of disjoint subsets are additive, $\mu_a(X)=\mu_a(A)+\mu_a(X\setminus A)$ for $a\in\{1,2\}$.  Hence, $\mu_1(A)=\mu_1(X)=1\implies\mu_1(X\setminus A)=0$. From absolute continuity $\mu_1(X\setminus A)=0\implies\mu_2(X\setminus A)=0$. Therefore $\mu_2(A)=\mu_2(X)=1$.
	\end{proof}
	Under any measure $d\mu$, a set $D\subset X$ of full measure need not be the entire set or even dense in $X$. 
	A measure is \emph{strictly positive} if it assigns non-zero measure to every open subset of $X$. Therefore, a set of full measure under a strictly positive measure is automatically dense.
	\section{Proof of main results} 
	\label{app:proofs}
	\subsection{Proof of proposition~\ref{prop:toy_reachability}}
	\label{app:toy_reachability}
	Note that the function $\Gamma^*:\Theta\to\mc{O}_{\psi_r}$ is smooth and  when $\max_{\vtheta\in\Theta}\rank(d\Gamma^{*}_\vtheta)=\max_{\vtheta\in\Theta}\rank(\mbf{F}^*(\vtheta))<\dim\mc{O}_{\psi_r}$,  it is nowhere locally surjective. Thus, the entire image $\im(\Gamma^{*})$ consists of \emph{critical values} and therefore $P\left(\psi_*\in\im(\Gamma^{\psi_r})\right)=0$ via theorem~\ref{theorem:sard}. On the other hand,  if $\max_{\vtheta\in\Theta}\rank(\mbf{F}^*(\vtheta))=\dim\mc{O}_{\psi_r}$,  $\im(\Gamma^*)$ contains an open neighbourhood in $\mc{O}_{\psi_r}$ via the \emph{local submersion theorem}. Since $d\Omega_{\psi_r}$ is a Riemannian volume form and assigns non-zero volume to every open subset in $\mc{O}_{\psi_r}$ (i.e. it is a \emph{strictly positive} measure), we have $P\left(\psi_*\in\im(\Gamma^{\psi_r})\right)>0$. 
	
	\subsection{Proof of lemma~\ref{lemma:surj_equiv}}
	\label{app:surj_equiv} 
	The openness of $\M_0$ in $\M_N$ makes the preimage $f^{-1}(\M_0)=f^{-1}(\im(f)\cap\M_0)$ an open set in $\M^*$ from the continuity of $f$. Therefore the restriction $\tilde{\Sigma}$ is smooth. The forward implication follows since $\M_0$ is dense in $\M_N$ and every open set in $\M_N$ contains an element of $\M_0$. The reverse implication follows by inclusion, since $(g',p')$ already belongs to $V$ by definition. 
	\subsection{Proof of Theorem~\ref{theorem:suff_nec}} 
	\label{app:suff_nec}
	Let $\M_0$ denote the union of all principal orbits in $\M_N$. From the principal orbit-type theorem, $\M_0$ is an open subset in $\M_N$. Therefore, under the subspace topology in $\im(f)$ the intersection of $\M_0$ with $\im(f)$ is an open set. Note that the intersection can be empty for which $\im(f)$ belongs to the complement $\bar{\M}_0=\M_N\setminus\M_0$. 
	Nevertheless, $\im(f)$ partitions into two subsets $S_0=\im(f)\cap\M_0$ and $\bar{S}_0=\im(f)\cap\bar{\M_0}$. Accordingly we can define $\mc{K}_*^0=\bigcup_{\psi^*\in S_0}\mc{O}_{0,\psi^{*}}$ made up of principal orbits, and $\mc{\bar{K}}_*^0=\bigcup_{\psi^*\in \bar{S}_0}\mc{O}_{\psi^{*}}$ with exceptional or singular orbits. Moreover, since distinct orbits are mutually exclusive, these two sets also partition the set $\mc{K}_*$, i.e. $\mc{\bar{K}}_*^0\cup\mc{K}_*^0=\mc{K}_*$ and $\mc{\bar{K}}_*^0\cap\mc{K}_*^0=\emptyset$. Thus,  we have 
	\begin{equation*}
		\Omega_{\rm FS}(\mc{K}_*)=\Omega_{\rm FS}(\mc{K}_*^0)+\Omega_{\rm FS}(\mc{\bar{K}}_*^0)\quad.
	\end{equation*}
	
	For the case of proper action on compact manifolds, $\bar{\M}_0$, which is the union of all singular and exceptional orbits has $\dim(\bar{\M}_0)\leq\dim(\M_N)-1$ (see theorem 3.8 in ref.~\cite{BredonGlenE1972Itct} for an extensive discussion and proof). 
	Therefore, $\mc{\bar{K}}_*^0$ which is the intersection of $\bar{\M}_0$ with $\mc{K}_*$  is  a null set under $d\Omega_{\rm FS}$, i.e. $\Omega_{\rm FS}(\mc{\bar{K}}^0_*)=0$.
	The forward implication i.e. local surjectivity at one point implies $\Omega_{\rm FS}(\mc{K}_*)>0$ since $d\Omega_{\rm FS}$ is strictly positive. For the reverse implication, from lemma~\ref{lemma:surj_equiv}, since local surjectivity of the restriction is equivalent to that of the full map, 
	when $\tilde{\Sigma}$ is nowhere locally surjective, $\Sigma$ is nowhere locally surjective. Hence, every elements iof $\im(\Sigma)$ is a critical values, $\im(\Sigma)$ therefore has zero measure from theorem~\ref{theorem:sard}.

	\subsection {Proof of Theorem ~\ref{theorem:gen_analytic}}
	\label{app:gen_analytic}
	Define restrictions of the map $\Sigma$ to either factors as $\Sigma^g:\M^*\to\M_N$ and $\Sigma^p:\G\to\M_N$ with the assignment $p\mapsto\Sigma(g,p)$ and $g\mapsto\Sigma(g,p)$ so that $\rank(d\Sigma^p_g)=\rank(d\Sigma^p_e)=\dim\mc{O}_{\Sigma(e,p)}\leq\dim\mc{O}_0$ and $\rank(d\Sigma^g_p)=\dim(\M^*)$.
	
	Noting the Jacobian matrix of $d\Sigma^p_e$ depends analytically on $p$ in the connected domain $\M^*$ and $S_0$ being non-empty, it achieves the maximum possible rank (from lemma~\ref{lemma:analyt_gen_rank}), we get $\Omega_{\im(f)}(S_0\cup S_E)=1$ under the Riemannian volume form $d\Omega_{\im(f)}$ on $\im(f)$. This is because measurable subsets of  $\M^*$ are pushed forward with preserved measures under the embedding $f$, while maximal rank is possible on exceptional orbits as well, i.e. those orbits with the same topological dimensions as the principal orbits but are not diffeomorphic.  From lemma~\ref{lemma:conull_abs_cont}, we get $\mu_*(S_0\cup S_E)=1$ for any absolutely continuous measure $d\mu_*$ with respect to $d\Omega_{\im(f)}$.

	First, we see that $\rank(d\Sigma_{(e,p)})=\rank(d\Sigma_{(g,p)})$ for all $g\in\G$ as we can always translate the map diffeomorphically to any element of the connected Lie group $\G$ via left translation. In other words, the family $\Sigma^g:\M^*\to\M_N$ parametrises a family of smooth embeddings of $\M^*$. The assumption $\dim \mc{O}_0+\dim\M^*>\dim\M_N$ makes $r_{\rm max}=\dim\M_N$.
	Therefore, applying lemma~\ref{lemma:analyt_gen_rank} to the Jacobian matrix of $d\Sigma_{(e,p)}:T_e\G\oplus T_p\M^*\to T_{\Sigma(e,p)}\M_N$, we get $R_0\neq\emptyset\implies \mu_*(R_0)=1$. Since $R_0$ is also a set of full measure with the \emph{strictly positive} Riemannian volume form in $\im(f)$, it is also dense in $\im(f)$. 
	
	To see that $S_0$ is dense in $\im(f)$, it suffices to show that an open neighbourhood that contains an element of $R_0$ contains an element of $S_0$ since $R_0$ is already dense in $\im(f)$. The equivalence of local surjectivity of $\Sigma$ and $\tilde{\Sigma}$ from lemma~\ref{lemma:surj_equiv} guarantees that any open subset that contains an element of $R_0$ contains an element of $S_0$. Therefore, $S_0$ is dense in $\im(f)$.

	\bibliographystyle{apsrev4-2}
	\bibliography{ref}

\end{document}